\documentclass{llncs}
\usepackage{etex}

\usepackage{amsmath,amsfonts,amssymb}
\usepackage{graphicx}
 \pdfoutput=1
\usepackage{latexsym}
\usepackage[all]{xy}
\usepackage{tikz}

\makeatletter
\let\orgdescriptionlabel\descriptionlabel
\renewcommand*{\descriptionlabel}[1]{%
  \let\orglabel\label
  \let\label\@gobble
  \phantomsection
  \edef\@currentlabel{#1}%
  \let\label\orglabel
  \orgdescriptionlabel{#1}%
}
\makeatother

\newtheorem{fact}{Fact}

\newcommand{\Log}{\mbox{{\sf L}}}

\newcommand{\FL}{\mbox{\sf FL}}

\newcommand{\Ceql}{\mbox{{\sf C=L}}}
\newcommand{\ShP}{\mbox{{\sf \#P}}}

\newcommand{\NP}{\mbox{{\sf NP}}}
\newcommand{\Gl}{\mbox{{\sf GapL}}}

\newcommand{\ACz}{\mbox{{\sf AC}$^0$}}

\newcommand{\poly}{\mbox{{\sf poly}}}

\newcommand{\calC}{\mbox{${\cal C}$}}

\newcommand{\parL}{\mbox{{\sf $\oplus$L}}}

\newcommand{\ION}{{\mathsf{NXT}}}

\newcommand{\mat}[4]{\left( \begin{array}{cc}
							       #1 & #2 \\
							       #3 & #4
                                  \end{array} \right)}

\begin{document}
\title{Bounded Treewidth and Space-efficient Linear Algebra}
\author{Nikhil Balaji \& Samir Datta\thanks{Part of the work was done on a visit to the 
Institute for Theoretical Computer Science at Leibniz University Hannover}}
\institute{
  Chennai Mathematical Institute (CMI), India\\
  \email{\{nikhil,sdatta\}@cmi.ac.in}
}
\maketitle
\begin{abstract}
Motivated by a recent result of Elberfeld, Jakoby and Tantau\cite{EJT} 
showing that $\mathsf{MSO}$ properties are Logspace computable on graphs 
of bounded tree-width, we consider the complexity of computing the determinant 
of the adjacency matrix of a bounded tree-width graph and as our main result 
prove that it is in Logspace. It is important to notice that the determinant 
is neither an $\mathsf{MSO}$-property nor counts the number of solutions of an 
$\mathsf{MSO}$-predicate.  This technique yields Logspace algorithms for 
counting the number of spanning arborescences and directed Euler tours in 
bounded tree-width digraphs.

We demonstrate some linear algebraic applications of the determinant algorithm 
by describing Logspace procedures for the characteristic polynomial,
the powers of a weighted bounded tree-width graph and feasibility of a system 
of linear equations where the underlying bipartite graph has bounded 
tree-width.

Finally, we complement our upper bounds by proving \Log-hardness of the 
problems of computing the determinant, and of powering a bounded tree-width 
matrix. We also show the \Gl-hardness of Iterated Matrix Multiplication where	
each matrix has bounded tree-width.
\end{abstract}

\section{Introduction}
The determinant is a fundamental algebraic invariant of a matrix.
For an $n \times n$ matrix $A$ the determinant is given by the expression
\(
 \mbox{Det}(A) = \sum_{\sigma \in S_n} \mbox{sign}(\sigma) \prod_{i \in [n]} a_{i,\sigma(i)}
\)
where $S_n$ is the symmetric group on $n$ elements, $\sigma$ is a permutation
from $S_n$ and $\mbox{sign}(\sigma)$ is the parity of the number of inversions in $\sigma$
($\mbox{sign}(\sigma) = 1$ if the number of inversions in $\sigma$ is even and $0$ if it is odd). 
Even though the summation in the definition runs over $n!$ many terms, there are many
efficient sequential \cite{vzGG} and parallel \cite{Berk} algorithms for computing the determinant.

Apart from the inherently algebraic methods to compute the determinant there
are also combinatorial algorithms (see, for instance, Mahajan and Vinay \cite{MV}) 
which extend the definition of determinant as a signed sum of cycle covers in the weighted
adjacency matrix of a graph. \cite{MV} are thus able to give another proof
of the \Gl-completeness of the determinant, a result first proved by 
Toda \cite{Toda91}.
For a more complete discussion on the known algorithms for the determinant, see \cite{MV}.


Armed with this combinatorial interpretation of the determinant
 and faced with its \Gl-hardness, 
one can ask if the determinant is any easier when the underlying matrix represents simpler
classes of graphs. Datta, Kulkarni, Limaye, Mahajan \cite{DKLM} study the complexity of the 
determinant and permanent, when the underlying directed graph is planar and show that they are as 
hard as the general case - $\Gl$ and $\ShP$-hard, respectively.
We revisit these questions in the context of bounded tree-width graphs.
%

Many $\NP$-complete graph problems become tractable when restricted to 
graphs of bounded tree-width. In an influential paper, Courcelle \cite{Cou} proved
that any property of graphs expressible in Monadic Second Order $\mathsf{MSO}$ logic 
can be decided in linear time on bounded tree-width
graphs. For example, Hamiltonicity is an $\mathsf{MSO}$ property and hence deciding if a 
bounded tree-width graph has a Hamiltonian cycle can be done in linear time.
More recently Elberfeld, Jakoby, Tantau \cite{EJT} showed that in fact, $\mathsf{MSO}$ 
properties  on bounded tree-width graphs can be decided in \Log. 

We study the Determinant problem when the underlying directed graph has 
bounded tree-width and show a Log-space upper bound. 
In the same vein we also compute other linear algebraic invariants of a
bounded tree-width matrix, such as the characteristic polynomial, rank and powers 
of a matrix in Logspace. Interpreting rectangular matrices as (weighted) 
bipartite graphs, we are also able to show that checking for the feasibility
of a system of linear equations for such matrices arising from bounded tree-width
bipartite graphs is in \Log. $\mathsf{FSLE}$ has previously been studied for general graphs
in \cite{ABO} where it is shown to be complete for the first level of the Logspace
counting hierarchy: $\Log^{\Ceql}$.

We give a tight bound on the complexity of the determinant by showing 
that it is $\Log$-hard via a reduction from directed reachability in paths. 
We also show that it is unreasonable to attempt to extend the Logspace upper bound of
determinant and powering to Iterated Matrix Multiplication (IMM) of bounded tree-width 
matrices, by showing \Gl-hardness for IMM. It is worthwhile to contrast this  with the
case of general graphs, where the Determinant, IMM and Matrix Powering are known to be 
inter-reducible to each other and hence complete for \Gl.

Counting spanning trees in directed graphs is easily  seen to be in \Gl\ by the
matrix tree theorem\cite{RPSalgebraic}. Counting modulo $2$
 has recently been proved to be \parL-hard for planar graphs in \cite{DK}.
As a direct consequence of our determinant 
result and the Kirchoff matrix tree theorem it follows that the problem is in \Log\ 
for graphs of bounded tree-width.

The {\bf BEST} theorem due
to De Bruijn, Ehrenfest, Smith and Tutte gives an exact formula for
the number of Euler tours in a directed graph (see Fact \ref{fact:best})
in terms of the number of directed spanning trees of the graph.

\subsection{Our Results and Techniques}
 Through out this paper, we work with matrices with entries from $\mathbb{Q}$,
 unless stated otherwise. We show that the following can be computed/tested in \Log:
\begin{enumerate}
 \item (Main Result)
 The Determinant of an $(n \times n)$ matrix $A$ whose underlying 
 undirected graph has bounded tree-width. As a corollary we can also compute
 the coefficients of the characteristic polynomial of a matrix.
 \item The inverse of an $(n \times n)$ matrix $A$ whose underlying 
 undirected graph has bounded tree-width. As a corollary we get a
 Logspace algorithm to compute the powers $A^k$ of a matrix $A$ 
 (with rational entries) whose support is a bounded tree-width digraph.
  \item Testing if a system of rational linear equations $Ax = b$
  is feasible where $A$ is (a not necessarily square) matrix 
  whose support is the biadjacency matrix of an undirected bipartite
  graph of bounded tree-width.
 \item The number of Spanning Trees in graphs of bounded tree-width.
 \item The number of Euler tours in a bounded tree-width directed graph
\end{enumerate}

We also show hardness results to complement the 
above easiness results:
\begin{enumerate}
 \item Computing the determinant of a bounded tree-width matrix is 
\Log-hard which precludes further improvement in the Logspace upper bound.
\item Computing the iterated matrix multiplication of bounded tree-width
matrices is \Gl-hard which precludes attempts to extend the \Log-bound
on powering matrices of bounded tree-width to iterated matrix 
multiplication.
\item Powering matrices are however \Log-hard which prevents attempts to
further improve the \Log-bound on matrix powering.
\end{enumerate}

At the core of the results is our algorithm to compute the determinant
by writing down an $\mathsf{MSO}$ formula that evaluates to true on
every valid cycle cover of the bounded tree-width graph underlying $A$. 
The crucial point being that the cycle covers are parameterised on the
number of cycles in the cycle cover, a quantity closely related to the 
sign of the cycle covers. This makes it possible
to invoke the cardinality version of Courcelle's theorem(for Logspace) due 
to \cite{EJT} to compute the determinant. A more subtle point is that in
order to keep track of the number of cycles as the size of a set of vertices,
we need to pick one vertex per cycle. Picking one vertex per cycle is done
by choosing the ``smallest'' vertex in the cycle. In order to pick a vertex
in a cycle cover, we need to define a total order on the vertices which
makes this part of the proof technically challenging.

We use this determinant algorithm  and the Kirchoff matrix tree theorem
along with the BEST theorem to count directed Euler tours. 


\subsection{Organization of the paper}
Section \ref{sec:prelim} introduces some notation and terminology required for the rest of the paper.
In Section \ref{sec:btwDet}, we give a Logspace algorithm to compute the Determinant of 
matrices of bounded tree-width and give some linear algebraic and graph theoretic applications. 
In Section~\ref{sec:hardness}, we give some 
$\Log$-hardness results to complement our Logspace upperbounds.
In Section \ref{sec:concl}, we mention some problems that remain open.

\section{Preliminaries}\label{sec:prelim}

\subsection{Background on Graph Theory}
\begin{definition}\label{def:tdecomp}
Given an undirected graph $G = (V_G, E_G)$ a tree decomposition of $G$ is a tree
$T = (V_T, E_T)$(the vertices in $V_T \subseteq 2^{V_G}$ are called \textit{bags}), 
such that 
\begin{enumerate}
 \item Every vertex $v \in V_G$ is present in at least one bag, i.e., 
 $\cup_{X \in V_T} X = V_G$.
 \item If $v \in V_G$ is present in bags $X_i, X_j \in V_T$, then
 $v$ is present in every bag $X_k$ in the unique path between $X_i$
 and $X_j$ in the tree $T$.
 \item For every edge $(u, v) \in E_G$, there is a bag $X_r \in V_T$ such that
 $u, v \in X_r$.
\end{enumerate}
The width of a tree decomposition is the $\max_{X \in V_T} (|X| - 1)$. The tree width of
a graph is the minimum width over all possible tree decomposition of the graph.
\end{definition}

\begin{definition}\label{def:ccover}
Given a weighted directed graph $G = (V, E)$ by its adjacency matrix $[a_{ij}]_{i,j\in[n]}$,
a \textit{cycle cover} $\calC$ of $G$ is a set of vertex-disjoint cycles that cover the vertices of $G$. 
I.e., $\calC = \{C_1, C_2, \ldots, C_k\}$, where $V(C_i) = \{c_{i_1}, \ldots, c_{i_r}\} \subseteq V$ such that
$(c_{i_1}, c_{i_2})$, $(c_{i_2}, c_{i_3})$, $\ldots$, $(c_{i_{r-1}}, c_{i_r})$, $(c_{i_r}, c_{i_1}) \in E(C_i) \subseteq E$ and 
$\sqcup_{i=1}^k V(C_i) = V$. 
\end{definition}

\begin{fact}\label{fact:ccsign}
The weight of the cycle $C_i = \prod_{j \in [r]} \mbox{wt}(a_{ij})$
and the weight of the cycle cover $\mbox{wt}(\calC) = \prod_{i \in [k]} \mbox{wt}(C_i)$. The sign of
the cycle cover $\calC$ is $(-1)^{n+k}$. 
\end{fact}


Every permutation $\sigma \in S_n$ can be written as a union of vertex disjoint cycles. Hence a permutation  
corresponds to a cycle cover of a graph on $n$ vertices. In this light, the determinant of an $(n \times n)$
matrix $A$ can be seen as a signed sum of cycle covers:
\[
 \mbox{det}(A) = \sum_{\mbox{cycle cover\ } \calC} \mbox{sign}(\calC) \mbox{wt}(\calC)
\]

%
%

\subsection{Background on $\mathsf{MSO}$-logic}

\begin{definition}[Monadic Second Order Logic]\label{def:mso}
 Let the variables $V = \{v_1, v_2, \ldots, v_n\}$ denote the vertices of a graph $G = (V, E)$.
 Let $X, Y$ denote subsets\footnote{The case when quantification over subset of edges
 is not allowed is referred to as $\mathsf{MSO}_1$ which is known to be strictly less powerful
 than $\mathsf{MSO}_2$, the case when edge set quantification is allowed. Throughtout our paper, we will 
 work with $\mathsf{MSO}_2$ and hence we will just refer to it as $\mathsf{MSO}$.} of $V$ or $E$.
 Let $E(x,y)$ be the predicate that evaluates to $1$ when there 
 is an edge between $x$ and $y$ in $G$. A logical formula $\phi$ is called an $\mathsf{MSO}$-formula if it can be 
 constructed using the following:
 \begin{itemize}
  \item $v \in X$
  \item $v_1 = v_2$
  \item $E(v_1,v_2)$
  \item $\phi_1 \lor \phi_2$, $\phi_1 \land \phi_2$, $\lnot \phi$
  \item $\exists x \phi, \forall x \phi$
  \item $\exists X \phi, \forall X \phi$  
 \end{itemize}

 In addition, if the Gaifman graph\footnote{The Gaifman graph (also called the \emph{Primal Graph}) of
 a binary relation $R \subseteq A \times A$ is the graph whose nodes are elements of $A$  
 and an edge joins a pair of variables $x,y$ if $(x,y) \in R$.}
 of the relation is bounded treewidth then we can use any predicate in 
 item 3 above. A property $\Pi$ of graphs is $\mathsf{MSO}$-definable, if it can be expressed as a 
 $\mathsf{MSO}$ formula $\phi$ such that $\phi$ evaluates to TRUE on a graph $G$ if and only if $G$ has property $\Pi$.
\end{definition}

\begin{definition}[Solution Histogram]
 Given a graph $G = (V, E)$ and an $\mathsf{MSO}$ formula $\phi(X_1, \ldots, X_d)$ in free 
 variables $X_1, \ldots, X_d$, where $X_i \subseteq V$(or $E$), the $(i_1, \ldots, i_d)$-th
 entry of $\mbox{histogram}(G, \phi)$ gives the number of subsets $S_1, \ldots, S_d$ such
 that $|S_j| = i_j$ for which $\phi(S_1, \ldots, S_d)$ is true. 
\end{definition}

We need the following results from \cite{EJT}:

\begin{theorem}[Logspace version of Bodlaender's theorem]
\label{prop:EJTB}
For every $k \ge 1$, there is a Logspace machine that on input of any 
graph $G$ of tree width at most $k$ outputs a width-$k$ tree decomposition of $G$.
\end{theorem}

\begin{theorem}[Logspace version of Courcelle's theorem]
\label{prop:EJTC}
For every $k \ge 1$ and every $\mathsf{MSO}$-formula $\phi$, there is a  
Logspace machine that on input of any logical structure $\mathcal{A}$ of tree width at
most $k$ decides whether $A \vDash \phi$ holds.
\end{theorem}

\begin{theorem}[Cardinality version of Courcelle's theorem]\label{thm:EJTcard}
Let $k \geq 1$ and let $\phi(X_1, \ldots , X_d)$ be an $\mathsf{MSO}$-formula on
free variables $X_1, \ldots , X_d$. Then there is a Logspace machine that on input of 
the tree decomposition of a graph $G$ of treewidth at most $k$, $\mathsf{MSO}$-formula
$\phi$ and $(i_1, \ldots, i_d)$, outputs the value of $\mbox{histogram}(G, \phi)$ at 
$|X_1| = i_1, \ldots, |X_d| = i_d$.
\end{theorem}

\section{Determinant Computation} \label{sec:btwDet}
Given a square $\{0,1\}$-matrix $A$, we can view it as the bipartite adjacency
matrix of a bipartite graph $H_A$. The permanent of this matrix $A$ counts the 
number of perfect matchings in $H_A$, while the determinant counts the signed
sum of perfect matchings in $H_A$.

If $G$ is a bounded treewidth graph then we can count the number of perfect
matchings in $G$ in \Log \cite{EJT} (see also \cite{DDN}). Hence the 
complexity of the permanent of $A$, above is well understood in this case 
while the complexity of computing the determinant is not clear. 

On the other hand the determinant of a $\{0,1\}$-matrix reduces (say by
a reduction $g_{MV}$) to counting
the number of paths in another graph (see e.g. \cite{MV}). Also counting
$s,t$-paths in a bounded treewidth graph is again in $\Log$\ via \cite{EJT}
(see also \cite{DDN}). But the problem with this approach is that that the
graph $g_{MV}(G)$ obtained by reducing a bounded treewidth $G$ is not
bounded treewidth.

However, we can also view $A$ as the adjacency matrix of a directed
graph $G_A$. If $G_A$ has bounded treewidth (which implies that $H_A$ also 
has bounded treewidth, see Proposition~\ref{prop:btwSplit} in Appendix~\ref{sec:split}) 
then we have a way of computing the determinant of $A$. 
To see this, consider the following lemma:

\begin{lemma} \label{lem:thetaX}
There is an $\mathsf{MSO}$-formula $\phi(X, Y)$ with 
free variables $X, Y$ that take values from the set of subsets of vertices
and edges respectively, such that $\phi(X, Y)$ is true exactly when $X$ is 
the set of heads of a cycle cover $Y$ of the given graph.
\end{lemma}

Before proving this Lemma we need some preprocessing.
Let $G$ be the input graph of bounded tree-width. We will augment 
$G$ with some new vertices and edges to yield a graph $G'$ again
with a tree decomposition $T'$ of bounded tree-width. Then we have:
\begin{lemma}\label{lem:mainION}
There exists a relation $\ION$ on vertices of $G'$ which satisfies the 
following:
\begin{enumerate}
\item $\ION$ is compatible\footnote{Binary relation $R$ is said to be compatible with the tree 
decomposition $T'$ of $G$ if the Gaifman graph of $R$  
has $T'$ as its tree decomposition.} with the tree decomposition $T'$
\item $\ION$ is a partial order on the vertices of $G'$
\item $\ION$ is computable in \Log
\item The transitive closure $\ION^{*}$ is a total order when restricted
to the vertices of $G$
\item $\ION^{*}$ is expressible as an $\mathsf{MSO}$-formula over the
vocabulary of $G'$ along with $\ION$.
\end{enumerate}
\end{lemma}

The construction of such a relation is fairly straight forward and considered folklore in
the Finite Model Theory literature (See for example Proposition VI.4 in ~\cite{CF12}. Here
we include an proof of Lemma~\ref{lem:mainION} (obtained independently) in the appendix for 
the sake of completeness.
\begin{proof}{(of Lemma \ref{lem:thetaX})}
We write an $\mathsf{MSO}$ formula $\phi$ on free variables $X, Y$, such that
$Y \subseteq E$ and $X \subseteq V$, such that $\phi$ evaluates to true on
any set of heads of a cycle cover $S$. The $\mathsf{MSO}$ predicate essentially  
verifies that the subgraph induced by $Y$ indeed forms a cycle cover of $G$.
Our $\mathsf{MSO}$ formula is of the form
\footnote{Note that since we require that for a given $X,Y$, 
every $v \in V$ has a unique $h \in X$, our formula is not monotone, i.e., If $X \subseteq
X'$ are two sets of heads then if $\phi(X,Y)$ is true doesn't imply $\phi(X',Y)$ is 
also true (consider vertices in $X'\setminus X$, since $X' \subseteq X$, they will have
two different $h, h'$ such that the $\mathsf{PATH}$ and $\ION^{*}$ predicates are true
contradicting uniqueness of $h$}: 
\[
\phi(X,Y) \equiv  (\forall v \in V) (\exists! h \in X) [\mathsf{DEG}(v,Y)
 \wedge \mathsf{PATH}(h,v,Y) \wedge (\ION^{*}(h,v) \lor \mathsf{EQ}(h,v))] 
\]
where,
\begin{enumerate}
 \item $\mathsf{DEG}(v,Y)$ is the predicate that says that the in-degree and out-degree
 of $v$ (in the subgraph induced by the edges in $Y$) is $1$.
 \item $\mathsf{PATH}(x, y,Y)$ is the predicate that says that there is a
  path from $x$ to $y$ in the graph induced by edges of $Y$.
 \item $\mathsf{EQ}(h,v) = 1$ iff $h = v$
\end{enumerate}
One can check that all the predicates above are $\mathsf{MSO}$-definable. 
%
\end{proof}
Lemma \ref{lem:thetaX} along with the Fact~\ref{fact:ccsign} yields:
\begin{lemma}\label{lem:ratdet}
 Given an $(n \times n)$ bounded treewidth matrix $A$ with integer entries, there is a Logspace
 algorithm that constructs an $(m \times m)$(where $m = \poly(n)$) matrix $B$ with entries from $\{0,1\}$,
such that $\mbox{det}(A) = \mbox{det}(B)$ and the treewidth of $B$ is the same as
 the treewidth of $A$.
\end{lemma}
Thus, using the histogram version of Courcelle's theorem from \cite{EJT} and 
Lemma~\ref{lem:ratdet}, we get:
\begin{theorem}\label{thm:btwDet}
The determinant of a matrix $A$ with integer entries, which can be viewed as the 
adjacency matrix of a weighted directed graph of bounded treewidth, is in \Log.
\end{theorem}
\begin{proof}
Firstly, obtain the matrix $B$ from $A$ using Lemma~\ref{lem:ratdet}. The histogram version
of Courcelle's theorem as described in \cite{EJT} when applied to the formula $\phi(X,Y)$ above
yields the number of cycle covers of $G_B$ parametrized on $|X|, |Y|$. But in the notation of
Fact~\ref{fact:ccsign} above, $|X| = k$ and $|Y| = n$, so we can easily compute the determinant
as the alternating sum of these counts.
\end{proof}
\begin{corollary}{\label{cor:charpoly}}
There is a Logspace algorithm that takes as input a $(n \times n)$ bounded treewidth 
matrix $A$, $1^m$, where $1 \leq m \leq n$ and computes the coefficient of $x^m$ in 
the characteristic polynomial ($\chi_A(x) = det(xI - A)$) of $A$.
\end{corollary}
The characteristic polynomial of an $(n \times n)$ matrix 
$A$ is the determinant of the matrix $A(x) = xI - A$. We could use 
Theorem~\ref{thm:btwDet} to compute this quantity (since $A(x)$
is bounded treewidth, if $A$ is bounded treewidth). However, Theorem
~\ref{thm:btwDet} holds only for matrices with integer entries while
the matrix $A(x)$ contains entries in the diagonal involving the 
indeterminate $x$. 

We proceed as follows: In the directed graph corresponding to $A$, 
replace a self loop on a vertex of weight $x-d$ by a gadget of weight $-d$
in parallel with a self loop of weight $x$ (In the event that there is no self
loop on a vertex in $A$, add a self loop of weight $x$ on the vertex).
Replace the weights on the other edges according to the gadget in 
Lemma~\ref{lem:ratdet}. We have added exactly $n$ self loops, each of 
weight $x$ (for the original vertices of $A$). 

 We first consider a generalisation of the determinant of $\{0,1\}$-matrices of
 bounded tree-width viz. the determinant of matrices where the entries are from a 
 set whose size is a fixed universal constant and the underlying graph consisting
 of the non-zero entries of $A$ is of bounded tree-width.
\begin{lemma}\label{lem:detsup}
Let $A$ be a matrix whose entries belong to a set $S$ of
fixed size independent of the input or its length. If the underlying digraph 
with adjacency matrix $A'$, where $A'_{ij} = 1$ iff $A_{ij} \neq 0$, is of 
bounded tree-width then the determinant of $A$ can be computed in \Log.
\end{lemma}
\begin{proof}
Let $s = |S|$ be a universal constant, $S = \{c_1,\ldots,c_s\}$
 and let $\mbox{val}_i$ be the predicates that partitions the edges of 
$G$ according to their values i.e.  $\mbox{val}_i(e)$ is true iff the edge $e$ has 
value $c_i \in S$.
Our modified formula $\psi(X,Y_1,\ldots,Y_s)$ will contain $s$ unquantified 
new edge-set variables $Y_1,\ldots, Y_s$ along with the old vertex variable $X$,
and is given by:
\[
\forall{e \in E}{\left(\left(e \in Y_i\rightarrow \mbox{val}_i(e)\right) \wedge
(e \in Y \leftrightarrow \vee_{i=1}^s{(e \in Y_i)}\wedge \phi(X,Y)\right)}
\]
Notice that we verify that the edges in the set $Y_i$ belong to 
the $i^{\mbox{th}}$ partition and each eadge in $Y$ is in one of the $Y_i$'s.
The fact that the $Y_i$'s form a partition of $Y$ follows from the assumption
that $\mbox{val}_i(e)$ is true for exactly one $i \in [s]$ for any edge $e$.

To obtain the determinant we consider the histogram parameterised on the $s$
variables $Y_1,\ldots,Y_s$ and the heads $X$. For an entry indexed by
$x,y_1,\ldots,y_s$, we multiply the entry by 
$(-1)^{n+x} \prod_{i = 1}^s{{c_i}^{y_i}}$ and take a sum over all entries.
\end{proof}
In light of the Lemma above, we can compute the characteristic polynomial as follows:
\begin{proof}{(of Corollary~\ref{cor:charpoly})}
While counting the number of cycle covers with $k$
cycles, we can keep track of the number of self-loops
occurring in a cycle cover. It is easy to see that we can
obtain the coefficient of $x^r$ in the characteristic polynomial
from the histogram outlined in Lemma~\ref{lem:detsup}.
Hence we can also compute the 
characteristic polynomial in \Log. 
\end{proof}
\begin{corollary}{\label{cor:rank}}
There is a Logspace algorithm that takes as input a bounded treewidth $(n \times n)$ matrix $A$,
and computes the rank of $A$.
\end{corollary}
\begin{proof}
Compute the characteristic polynomial of $A$ and use the fact that
the rank of $A$ is a number $r$ such that $x^{n-r}$ is the smallest
power of $x$ wth a non-zero coefficient.  
\end{proof}
$\mathsf{FSLE}(A,b)$ is the following problem: Given a system of $m$ linear equations (with integer 
coefficients, w.l.o.g) in variables $z_1, \ldots, z_n$ and a target vector $b$, we want to check if 
there is a feasible solution to $Az = b$. That is, we want to decide if there is a setting of
the variable vector $z \in \mathbb{Q}^n$ such that, $Az = b$ holds for a bounded treewidth matrix 
$A \in \mathbb{Z}^{m \times n}$ (when we say a rectangular matrix is bounded treewidth, we mean
the underlying bipartite graph on $(m+n)$ vertices has bounded treewidth).
\begin{corollary}{\label{cor:fsle}}
 For a bounded treewidth matrix $A_{m \times n}$ and vector $b_{n \times 1}$, $\mathsf{FSLE}(A,b)$ is in $\Log$. 
\end{corollary}
\begin{proof}
 $A$ can be interpreted as the biadjacency matrix of a bipartite graph. Now,
 consider the matrix $B = \mat{0}{A}{A^T}{0}$ -- this is a matrix of dimension $(m+n)\times(m+n)$.
 It is easy to see that $B$ corresponds to the adjacency matrix of $A$. Let
 $\mbox{row-rank}(A) = \mbox{column-rank}(A) = r$. Since $A$ and $A^T$ have the same
 rank, $\mbox{rank}(A) + \mbox{rank}(A^T) = 2r = \mbox{rank}(B)$. Therefore, in order to
 find the rank of the rectangular matrix $A$, we can use the Logspace procedure for 
 matrix rank given by Corollary~\ref{cor:rank}. Now, we know that the system of linear
 equations given by $A, b$ is feasible if and only if $\mbox{rank}(A) = \mbox{rank}([A : b])$. 
\end{proof}
\begin{corollary}{\label{cor:inverse}}
There is a Logspace algorithm that takes as input a $(n \times n)$ bounded treewidth
matrix $A$, $1^i, 1^j, 1^k$ and computes the $k$-th bit of $A^{-1}_{ij}$.  
\end{corollary}
\begin{proof}
The inverse of a matrix $A$ is the matrix $B = \frac{\textbf{C}^T}{det(A)}$
where $\textbf{C} = (C_{ij})_{1 \leq i, j \leq n}$ is the cofactor matrix, whose $(i,j)$-th entry
$C_{ij} = (-1)^{i+j}\mbox{det}(A_{ij})$ is the determinant of the $(n-1)\times(n-1)$ matrix obtained
from $A$ by deleting the $i$-th row and $j$-th column. %
If we can compute $C_{ij}$ in $\Log$, we can compute the entries of $B$ via 
integer division which is known to be in $\Log$ from \cite{HAB}.
To this end, consider the directed graph $G_A$ represented by $A$. To compute
$\mbox{det}(A_{ij})$, swap the columns of $A$ such that the $j$-th column becomes the $i$-th
column. The graph so obtained is of bounded treewidth (To see this, notice that the swapping
operation just re-routes all incoming edges of $j$ to $i$ and those of $i$ to $j$.
The tree decomposition of this graph is just obtained by adding vertices $(i,j)$ to every bag
in the tree decomposition of $G_A$ and also the edges rerouted to the respective bags.
This increases the treewidth by $2$).
Now, remove the $i$-th vertex in $G_A$ and all edges incident to it to get a graph $G_{A_{ij}'}$ 
on $(n-1)$ vertices. The swapping operation changes the determinant of $A_{ij}$ by a sign that is 
$(-1)^{i-j} = (-1)^{i+j}$. Computing the determinant of this modified matrix $A_{ij}'$ yields $C_{ij}$ 
as required. Since $A_{ij}'$ is obtained from $A$ by removing a vertex and all the edges incident on it, 
the treewidth of $A_{ij}'$ is at most the treewidth of $A$. By Theorem \ref{thm:btwDet},
$C_{ij}$ is in $\FL$. 
\end{proof}

\begin{corollary}\label{cor:pow}
There is a Logspace algorithm that on input an $(n \times n)$ bounded treewidth matrix
$A$, $1^m, 1^i, 1^j, 1^k$ gives the $k$-th bit of $(i,j)$-th entry of $A^m$.
\end{corollary}

\begin{proof}
Consider $A' = (I - tA)^{-1}$ where $I$ is the $(n \times n)$ identity matrix and $t$
is a small constant to be chosen later. 
Notice that $A' = (I - tA)^{-1} = \sum_{j\geq 0} t^j A^j$.
By choosing $t$ as a suitably small power of $2$ (say $2^{-p} = t$ such that $2^p > \|A\|$) 
and computing $A'$ to a suitable accuracy, we can read the $(i,j)$-th entry of $A^m$ off the
appropriate bit positions of the $(i,j)$-th entry of $A'$. So, in essence the problem of powering
bounded treewidth matrix $A$ reduces to the problem of computing the inverse of a related matrix
which is known to be in $\Log\ $ via Corollary~\ref{cor:inverse}.
\end{proof}

\subsection{Spanning Trees and Directed Euler Tours}

%
\begin{fact}\label{fact:arb}
The number of arborescences of a digraph equals any
cofactor of its Laplacian.
\end{fact}
where the Laplacian of a directed graph $G$ is $D - A$ where $D$
is the diagonal matrix with the $D_{ii}$ being the out-degree of vertex $i$
and $A$ is the adjacency matrix of the underlying undirected graph.
The BEST Theorem states:
\begin{fact}[\cite{BE87}\cite{ST41}]\label{fact:best}
The number of Euler Tours in a directed Eulerian graph $K$ is exactly:
\[
t(K)\prod_{v \in V}{(\mbox{deg}(v) - 1)!}
\]
where $t(K)$ is the number of arborescences in $K$ rooted at an arbitrary
vertex of $K$ and $\mbox{deg}(v)$ is the indegree as well as the outdegree
of the vertex $v$.
\end{fact}
We combine Facts \ref{fact:arb} and \ref{fact:best}with Theorem~\ref{thm:btwDet}
to compute the number of directed Euler Tours in a directed Eulerian graph in
$\Log$.

%
Use the Kirchoff Matrix Tree theorem\cite{RPSalgebraic} and Fact~\ref{fact:best}:
\begin{corollary}\label{cor:btwDirEuler}
Counting arborescences and directed Euler Tours in a directed Eulerian graph $G$ 
(where the underlying undirected graph is bounded treewidth) is in \Log.
\end{corollary}

\section{Hardness Results}\label{sec:hardness}
We show a couple of hardness results to complement our Logspace
upper bounds. 
\begin{proposition}[Hardness of Bounded Treewidth Determinant]\label{lem:hard}
For all constant $k \ge 1$, computing the determinant of an $(n \times n)$ matrix $A$
whose underlying undirected graph has treewidth at most $k$ is $\Log$-hard.
\end{proposition}
\begin{proof}
 We reduce the problem $\mathsf{ORD}$ of deciding for a directed path $P$
 and two vertices $s, t \in V(P)$ if there is a path from $s$ to $t$ (known to be
 $\Log$-complete via \cite{Ete97}) to computing the determinant of bounded treewidth
 matrices (Note that $P$ is a path and hence it has treewidth $1$). 
%
 Our reduction is as follows: Given  a directed path $P$ with source $a$, sink $b$ 
 and distinguished vertices $s$ and $t$, we construct a new graph $P'$ as follows: 
 Add edges $(a, s'), (s', t), (t, s), (s, a)$ and  $(b, t')$ and remove edges 
 $(s', s), (t, t')$ where $s'$ and $t'$ are vertices in $P$ such that 
 $(s', s), (t, t') \in E(P)$ (See Figure~\ref{fig:st}).
 
 We claim that there is a directed path between $s$ and $t$ if and only if  
 the determinant of the adjacency matrix of $P'$ is zero.  
 If there is a directed path from $s$ to $t$ in $P$, then there are 
 two cycle covers in $P'$ : $(a,s') (s,t) (t',b)$, with three cycles and $(a, s', t, s)
 , (t', b)$, with two cycles. Using Fact \ref{fact:ccsign}, the signed sum of these cycle
 covers is $(-1)^{n+3} + (-1)^{n+2} = 0$, which is the determinant of $P'$.
 
 In the case that $P$ has a directed path from $t$ to $s$ (see Figure~\ref{fig:ts}), 
 then there is one cycle namely $(a, t, s', b, t', s)$. We argue as follows:
 The edges $(t,s), (s,b), (b,t'), (t',s')$ are in the
 cycle cover since they are the only incoming edges to $s, b, t', s'$ respectively. So $(t, s, b, t', s')$
 is a part of any cycle cover of the graph. This forces one to pick the edge $(s', a)$ and hence we have
 one cycle in the cycle cover for $P'$.
\end{proof}

%
%
%
%
\begin{figure}[ht]
\centering
\resizebox{0.6\textwidth}{!}{
\begin{tikzpicture}[every node/.style={circle, draw, scale=1.0, fill=gray!50}, scale=1.0, rotate = 180, xscale = -1]

\node (1) at (0, 0) {$a$};
\node (2) at (3.0, 0) {$s$};
\node (3) at (4.5, 0) {$s'$};
\node (4) at (7.5, 0) {$t'$};
\node (5) at (9.0, 0) {$t$};
\node (6) at (12.0, 0) {$b$};

\draw[->] (1) -- (2);
\draw[->] (2) -- (3);
\draw[->] (3) -- (4);
\draw[->] (4) -- (5);
\draw[->] (5) -- (6);
\end{tikzpicture}
} 

\resizebox{0.6\textwidth}{!}{
\begin{tikzpicture}[every node/.style={circle, draw, scale=1.0, fill=gray!50}, scale=1.0, rotate = 180, xscale = -1]

\node (1) at (0, 0) {$a$};
\node (2) at (3.0, 0) {$s$};
\node (3) at (4.5, 0) {$s'$};
\node (4) at (7.5, 0) {$t'$};
\node (5) at (9.0, 0) {$t$};
\node (6) at (12.0, 0) {$b$};

\draw[->] (1) -- (2);
\draw[thick, dashed, ->] (2) to[out=-100,in=-150] (1);
\draw[->] (3) -- (4);
\draw[thick, dashed, ->] (4) to[out=-100,in=-150] (3);
\draw[thick, dashed, ->] (2) to[out=-250,in=-350] (4);
\draw[thick, dashed, ->] (3) to[out=100,in=150] (1);	
\draw[->] (5) -- (6);
\draw[thick, dashed, ->] (6) to[out=-100,in=-150] (5);
\end{tikzpicture}
}
\caption{\label{fig:st} $s$ occurs before $t$}
\end{figure}
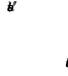
%

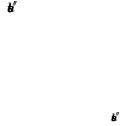
\begin{figure}[ht]
\centering
\resizebox{0.6\textwidth}{!}{
\begin{tikzpicture}[every node/.style={circle, draw, scale=1.0, fill=gray!50}, scale=1.0, rotate = 180, xscale = -1]

\node (1) at (0, 0) {$a$};
\node (2) at (3.0, 0) {$t$};
\node (3) at (4.5, 0) {$t'$};
\node (4) at (7.5, 0) {$s'$};
\node (5) at (9.0, 0) {$s$};
\node (6) at (12.0, 0) {$b$};

\draw[->] (1) -- (2);
\draw[->] (2) -- (3);
\draw[->] (3) -- (4);
\draw[->] (4) -- (5);
\draw[->] (5) -- (6);

\end{tikzpicture}
} 

\resizebox{0.6\textwidth}{!}{
\begin{tikzpicture}[every node/.style={circle, draw, scale=1.0, fill=gray!50}, scale=1.0, rotate = 180, xscale = -1]

\node (1) at (0, 0) {$a$};
\node (2) at (3.0, 0) {$t$};
\node (3) at (4.5, 0) {$t'$};
\node (4) at (7.5, 0) {$s'$};
\node (5) at (9.0, 0) {$s$};
\node (6) at (12.0, 0) {$b$};

\draw[->] (1) -- (2);
\draw[thick, dashed, ->] (4) to[out=-100,in=-150] (1);
\draw[->] (3) -- (4);
\draw[thick, dashed, ->] (5) to[out=100,in=150] (1);
\draw[thick, dashed, ->] (2) to[out=100,in=150] (5);
\draw[thick, dashed, ->] (4) to[out=-100,in=-160] (2);	
\draw[->] (5) -- (6);
\draw[thick, dashed, ->] (6) to[out=100,in=150] (3);
\end{tikzpicture}
}
\caption{\label{fig:ts}$t$ occurs before $s$}
\end{figure}
\begin{proposition}[Hardness of Bounded Treewidth Matrix Powering]\label{prop:powhard}
 Matrix Powering is $\Log$-hard under disjunctive truth table reductions.
\end{proposition}

\begin{proof}(of Proposition~\ref{prop:powhard})
We reduce $\mathsf{ORD}$ to matrix powering. Given an directed path $P$ on $n$ vertices 
and distinguished vertices $s$ and $t$, we argue as follows: There is a directed path between $s$
and $t$, then it must be of length $i$ for an unique $i \in [n]$. Consider the matrix $(I + A_P)^n$:
$s$ and $t$ are connected by a path if and only if $(I + A_P)^n_{s,t} \neq 0$. This is because
$(I + A_P)^n_{s,t}$ gives the walks from $s$ to $t$, and if at all there is a path from $s$ to
$t$, then there is definitely a walk of length at most $n$ between them. Checking if this entry is
zero can be done by a DNF which takes as input the bits of $(I + A_P)^n_{s,t}$. 
\end{proof}




 \begin{proposition}[Hardness of Bounded Treewidth IMM]\label{prop:imm}
  Given a sequence of bounded treewidth matrices with rational entries $M_1, M_2, \ldots, M_n$
  and $1^i, 1^j, 1^k$ as input, computing the $k$-th bit of $(i,j)$-th entry of $\prod_{l=1}^n
  M_l$ is $\Gl$-hard.
 \end{proposition}

\begin{proof}(of Proposition~\ref{prop:imm})
We reduce integer matrix powering to iterated matrix multiplication of bounded treewidth
graphs. Given an $(n \times n)$ matrix $A$ with $n$-bit entries, and $1^m, 1^i, 1^j$ the matrix
powering problem is to find the $(i,j)$-th entry of $A^m$. From the underlying digraph 
$G_A = (V = \{v_1,v_2, \ldots, v_n\}, E)$, we construct a sequence of bounded treewidth matrices as follows: 
We construct two gadgets $\mathcal{V}_l$ and $\mathcal{V}_u$ -- Both are graphs on $2n^2$ vertices divided in to $n^2$
partitions where each partition is a copy of $V$ (such that there are no edges between vertices in the partition):
$U = \sqcup_{i=1}^n U_i$ and $L = \sqcup_{i=1}^n L_i$ where each $L_i = U_i = V$. We also have the edges between:
\begin{enumerate}
 \item $v_i \in U_j$ and $v_i \in U_{j+1}$
 \item $v_i \in L_j$ and $v_i \in L_{j+1}$
\end{enumerate}
for all $i \in [n], j \in [n-1]$.The edges are basically an identity perfect matching between 
$U_i$ and $U_{i+1}$ and also $L_i$ and $L_{i+1}$. Now we add edges in $\mathcal{V}_l$
and $\mathcal{V}_u$ according to edges present in $G_A$: If $(v_i, v_1), \ldots, (u,v_r)$ 
are edges in $G_A$ out of vertex $v_i$, then
\begin{enumerate}
 \item In $\mathcal{V}_l$, we add an edge between $v_i \in L_i$ to $v_1, \ldots, v_r \in U_{i+1}$.
 \item In $\mathcal{V}_u$,  we add an edge between $v_i \in U_i$ to $v_{i_1}, \ldots, v_{i_r} \in L_{i+1}$
\end{enumerate}
We now construct a walk gadget $\mathcal{W}$ using alternating copies of the vertex gadgets
$\mathcal{V}_l$ and $\mathcal{V}_u$. To raise $A$ to the $m$-th power, construct:
$\mathcal{V}_1, \ldots, \mathcal{V}_m$ where $\mathcal{V}_1 = \mathcal{V}_l$,
$\mathcal{V}_2 = \mathcal{V}_u$ and so on. Connect $\mathcal{V}_i$ and $\mathcal{V}_{i+1}$
by the following edges:
\begin{enumerate}
 \item $v_i \in U_n$ and $v_i \in U_1, \forall i \in [n]$.
 \item $v_i \in L_n$ and $v_i \in L_1, \forall i \in [n]$.
\end{enumerate}
where $L_n, U_n \in \mathcal{V}_j$ and $L_1, U_1 \in \mathcal{V}_{j+1}$ for $j \in [m-1]$.
Additionally if $(v_n, v_{k_1}), \ldots, (v_n, v_{k_q})$ are edges out of $v_n$, then
add those corresponding edges between $L_n \in \mathcal{V}_j$ and $U_1 \in \mathcal{V}_{j+1}$
if $\mathcal{V}_j$ is a $\mathcal{V}_l$ gadget. Otherwise $\mathcal{V}_j$ is a $\mathcal{V}_u$
gadget and hence add the corresponding edges between $U_n \in \mathcal{V}_j$ and $L_1 \in \mathcal{V}_{j+1}$.
%
It is easy to see that there is a bijection between walks of length $m$ in $G_A$ 
and paths of length $m$ in $\mathcal{W}$. The gadget $\mathcal{W}$ that results is of
constant treewidth.
\end{proof}

\section{Open Problems}\label{sec:concl}
 What is the complexity of other linear algebraic invariants such as
minimal polynomial of a bounded tree-width matrix?
What is the complexity of counting Euler Tours in undirected tours
in bounded treewidth graphs? On general graphs, this problem is known to be
$\ShP$-complete\cite{BW}. See \cite{CCMsp,CCMtw} for some recent progress on this problem.
\bibliography{skeleton}

\section*{Acknowledgement}
We would like to thank Abhishek Bhrushundi, Arne Meier, Rohith Varma and Heribert Vollmer 
for illuminating discussions regarding this paper. Special thanks are due to
Johannes K\"obler and Sebastian Kuhnert who were involved in the initial discussions
on the proof of Theorems~\ref{thm:btwDet},~\ref{cor:btwDirEuler};
to Stefan Mengel for proof reading the paper and finding a gap
in a previous ``proof'' of Theorem~\ref{thm:btwDet}; and to Raghav Kulkarni for suggesting
proof strategies for Corollary~\ref{cor:pow} and Lemma~\ref{lem:ratdet}; and to Sebastian Kuhnert for the proof of 
Proposition~\ref{prop:powhard}. Thanks are also due to anonymous referees for pointing out errors
in a previous version of the paper, for greatly simplifying the proof of Corollary~\ref{cor:fsle} and for pointing out the reference \cite{CF12}.
This research is partially funded by a grant from the Infosys Foundation.

\begin{appendix}
\section{Background}\label{sec:background}
\subsection{Background on Complexity Classes}
For all the standard complexity classes, we refer the reader to \cite{AB09}.
We define some non-standard complexity classes that we refer to here:

Let $f = (f_n)_{n \ge 0}$ where $f_n : \{0,1\}^n \to \mathbb{Z}$ be a family 
of integer valued functions. $f$ is in the complexity class $\Gl$ if and only
if there is a nondeterministic logspace machine $M$ such that for every $x$,
$f(x)$ equals the difference  between the number of accepting and rejecting 
paths of $M$ on input $x$.

$\Ceql$ is the class of languages $L$ for which there exists a $\Gl$ functions
$f$ that evaluates to zero on precisely those $x \in L$. The class $\Log^{\Ceql}$
consists of all languages decidable by a logspace machine at the base with oracle 
access to a $\Ceql$-complete language. The $\ACz$ hierarchy over $\Ceql$ is built
as an $\ACz$ circuit with gates which can make queries to a $\Ceql$ oracle and is 
known to collapse to $\Log^{\Ceql}$. For more information on the subtleties of 
implementing this hierarchy, see \cite{ABO}.

\section{Split Graph and the Biadjacency matrix}\label{sec:split}
We will often find it convenient to interpret an arbitrary matrix $A \in \mathbb{Q}_{n \times n}$ as 
representing both a weighted  directed graph $G_A$ and the weighted 
biadjacency matrix of an undirected bipartite graph $H_A$. We  
will exploit this duality between these two interpretations of a matrix to compute linear algebraic 
invariants.

Let $H_A = (R \cup C, E)$ be the bipartite graph on $2n$ vertices where $A_{ij}$ represents an edge
between the $i \in R$ to $j \in C$ of the appropriate weight. When we say $H_A$ has bounded treewidth, 
we mean that the underlying undirected graph has bounded treewidth. If we consider $A$ to represent a digraph
$G_A$ on $n$ vertices, we can construct the \textit{split} graph
of $G_A$ namely $\mbox{Split}(G_A)$ -- for each vertex $v$ of $G_A$, we create $v_{in}$ and $v_{out}$. If there is an edge 
between $u$ and $v$ in $G_A$, then we add an edge in $\mbox{Split}(G_A)$ between $u_{out}$ and $v_{in}$. For all vertices
$v$, we add an edge $(v_{in}, v_{out})$. $\mbox{Split}(G_A)$ is an undirected graph on $2n$ vertices.
We have the following easy proposition:

\begin{proposition}\label{prop:btwSplit}
$G_A$ has bounded treewidth, iff $\mbox{Split}(G_A)$ has bounded treewidth. 
\end{proposition}
\begin{proof}
We can construct the tree decomposition of $\mbox{Split}(G_A)$ from that of $G_A$ by 
introducing one new vertex for every vertex in the bag. This almost doubles the treewidth. The other 
directions follows from the fact that $G_A$ can be obtained as minor of $\mbox{Split}(G_A)$. 
\end{proof}

\section{Omitted Proofs from Section~\ref{sec:btwDet}}
\subsection{Proof of Lemma~\ref{lem:mainION}}

\begin{proof}
We first define the graph $G'$ in which we add three vertices $b_l,
b_0,b_r$ for every bag $B$ in the tree decomposition of $G$. 
We will use the following running example, let $A$ be a parent bag 
in the tree decomposition with left child $B$ and right child $C$. 
The scheme we propose adds vertices $a_l, a_0, a_r$,
$b_l, b_0, b_r$ and $c_l, c_0, c_r$, with $a_l,a_0,a_r,b_r,c_l \in A$,
$a_l,b_l,b_0,b_r \in B$ and $a_r, c_l,c_0,c_r \in C$.


We put an edge between every bag vertex
$b \in V(G') \setminus V(G)$ and every vertex $v$ sharing a bag with
$b$. Clearly this is a valid tree decomposition of tree-width 
which is at most $6$\footnote{Alternately, following scheme also works: Add all the
$6$ bag vertices($b_l, b_0, b_r$ and $c_l, c_0, c_r$) of the children to the parent bag $A$. 
Symmetrically, add the bag vertices $a_l, a_0, a_r$ to both $B$ and $C$. Note that this is 
a local operation and it increases the treewidth of every bag by atmost $9$: $6$ from the 
children and $3$ from the parent. The final graph $G'$ obtained this way has treewidth at
most $12$ more than $G$.} greater than that of $G$.

We now traverse the tree $T'$ using an Euler Traversal
(see \cite{CookMckenzie}).
Every internal bag of the the tree is visited thrice in an Euler traversal - first when 
the traversal visits the bag the first time, second after exploring its left child, and 
finally third when it is done exploring its right child. For a bag $B$ in the tree 
decomposition, if $b_l, b_0, b_r$ are the vertices added, the vertices are visited in the 
following order in the Euler traversal of the tree: First $b_l$ is visited, after which the
left child of $B$ is explored. Following this, the traversal returns to $B$ via $b_0$, after
which the right child of $B$ is explored. Finally we visit $B$ again via $b_r$ and proceed to
$B$'s sibling via $B$'s parent. 

Define $\ION_{Bag}$ to be the following ordering of the bag vertices:
$a_l < b_l < b_0 < b_r < a_0 < c_l < c_0 < c_r < a_r$. In the case when 
$b$ has children, they are explored between $b_l$ and $b_0$ (left child) and $b_0$ and 
$b_r$ (right child). 

Next we extend the relation $\ION_{Bag}$ to a relation $\ION_1$ on pairs
which include vertices of $G$.  
For every vertex $v$ occurring in bag $A$ if $v$ does not belong to
the left child $B$ of $A$ but belongs to the right child $C$
then add the tuple $a_0 < v$ to $\ION_1$
else we just add $a_l < v$. Symmetrically, we add $v < a_0$ or
$v < a_r$ depending on whether $v$ belongs to the left child $B$
but not to the right child $C$, or not. If it does not belong
to either child, we add the tuples $b_r < v$ and $v < c_l$.

Now consider the transitive closure $\ION_1^{*}$ of $\ION_1$. This may well
not be a total order on the vertices of $G'$. But we have the following:
\begin{claim}
Bag vertices are totally ordered under $\ION_1^{*}$. If two vertices 
of $G$ are incomparable under $\ION_1^{*}$, then there must exist
a common bag to which both must belong.
\end{claim}
\begin{proof}
That bag vertices are totally ordered follows from the definition of
Euler traversal of a tree. To see the other part of the claim let
$u,v \in V(G)$ be unordered by $\ION_1^{*}$. Let the least common
ancestor (LCA) bag of all the bags to which $u$ belong be $Q$ and the 
LCA bag of all bags containing $v$ be $R$ and further $P$ is the 
LCA bag of $Q,R$. From the assumption that $u,v$ do not belong to
the same bag not all three $P,Q,R$ can be the same, in particular
$Q,R$ are distinct. If $P$ is distinct from $Q,R$
then from $\ION_1^{*}$ we know that: $q_r < p_0 < r_l$ (assuming,
wlog that $Q$ is a left descendant and $R$ a right descendant of $P$).
Also, $u < q_r$ (or possibly even $u < q_0 < q_r$)
and $r_l < v$ (or even $r_l < r_0 < v$) are present in $\ION_1$. Thus
$u < v$ is present in $\ION_1^{*}$. This leaves the case when $P = Q$
(the other case, $P = R$, is symmetric). Let $P'$ be the bag 
containing $u$ which
is nearest along the tree to $R$. Without (much) loss of generality
suppose $R$ is a descendant of the left child of $P'$. Let $R'$
be the left child of $P'$ ($R'$ may be the same as $R$). 
From the fact that $R$ is the highest bag
containing $v$, we know that $v < r_r$. Because $u$ does
not occur in  either child of $P'$ we have that $r'_r < u$.
Now $r_r$ is either the same as $r'_r$ (if $R = R'$) or $r_r < r'_r$
(if $R$ is a proper descendant of $R'$). In either case we get 
$v < u$ in the transitive closure of $\ION_1^{*}$. Other cases are 
similar.  
\end{proof}
We can find $\ION_1^{*}$ in \Log\ by using \cite{EJT}.
From this we find pairs $u,v\in V(G)$ which are incomparable under 
$\ION_1^{*}$. Next we order any unordered $u,v$ in \Log\ by their
 binary values and add these tuples (i.e. one of $u < v$ and $v < u$
for every unordered $u,v$) to $\ION_1$ to yield $\ION$.
Notice that the vertices of $G$ are totally ordered under 
$\ION^{*}$ and because of the claim above $\ION$ is compatible with
the tree decomposition $T'$
\end{proof}
\subsection{Proof of Lemma~\ref{lem:ratdet}}

\begin{proof}
  We replace each edge $(u,v)$ in $G_A$ by a series-parallel graph with edge weights from $\{0,1\}$ 
 such that the sum of weights of all paths between $u$ and $v$ exactly equals the weight of $(u,v)$
 and the construction doesn't alter the sign of the cycle cover in which $(u,v)$ is present.
 
 We use a construction similar to \cite{Val79} -- If the weight of the edge $(u,v)$ is 
 an $n$-bit integer $w = w_nw_{n-1}\ldots w_1$, we create a gadget of weight $2^i$ 
 for each $w_i$ and add edges $(u,s_i)$ and $(t_i, v)$. We add self loops on all the vertices
 in this gadget except $u$ and $v$. We also ensure that all paths between $u$ and $v$ are of equal length 
 (say $l$). If $n'$ is the number of new vertices added, then a cycle cover of $G_B$ will consist of 
 the original cycles in $G_A$ along with the $n' - l$ self loops that result from the series-parallel
 graph now representing $(u,v)$. Since the sign of a monomial in the determinant is decided by the number
 of cycles (equivalently the number of heads in a cycle cover) corresponding to the permutation of the 
 monomial, if we can ensure that $n' - l$ is always even, then the sign of the cycle cover will be 
 $(-1)^k = (-1)^{n' - l + k} = (-1)^k$. This is easily taken care of by replacing $v$ by a new vertex $v'$,
 removing all the edges $(t_i, v)$, introducing edges $(t_i, v')$ and $(v', v)$.

 It remains to show that from the tree decomposition of $G_A$, we can obtain a tree decomposition of $G_B$ 
 in $\Log$. This is easily accomplished as follows: For an edge $(u,v)$ in $G_A$, find the highest bag $b$ 
 in the tree where it occurs and add a new bag $b_{uv}$ as a child to $b$ which just contains $u$ and $v$ and
 the edge $(u,v)$. Attach the tree decomposition of the series-parallel graph gadget corresponding to $(u,v)$
 in $G_B$ as a child of $b_{uv}$. It is easy to verify that what results is still a tree decomposition --
 Every edge in $B$ is present in some bag. If $(u,v)$ is present in bags $b_1$ and $b_2$, it is present in
 all the bags in the unique path between $b_1$ and $b_2$ in the tree. Here one has to argue two cases: If
 $(u,v)$ is an edge in the series-parallel gadget, then it is only present in the bags corresponding to the
 tree decomposition of the series parallel graph (children of $b_{uv}$). Else $(u,v)$ must have been an edge
 in $G_A$. In this case, after appending the tree decomposition of the series-parallel graph to $b_{uv}$, it 
 still stays a tree decomposition. The size of any bag in this decomposition is $\mbox{max}(k,2)$ (since any
 series-parallel graph has treewidth at most $2$). 
\end{proof}
\end{appendix}
\end{document}